\documentclass[12pt]{article}

\usepackage{amsmath,amssymb,amsfonts,amsthm,mathtools,dsfont,mathrsfs} 
\usepackage{bm,bbm}
\usepackage{booktabs}
\usepackage{color}
\usepackage[dvips,letterpaper]{geometry}
\usepackage{epsfig}
\usepackage{fullpage}
\usepackage{here}
\usepackage{graphicx,psfrag,epsf}
\usepackage{indentfirst}
\usepackage{lineno}
\usepackage{natbib}
\usepackage{setspace}
\usepackage{subfig}
\usepackage{times}
\usepackage{url}
\usepackage{verbatim}
\usepackage{multirow}
\usepackage[table,xcdraw]{xcolor}
\usepackage{adjustbox, dsfont}
\usepackage{rotating}
\usepackage{float}

\usepackage{pdfpages} 
\usepackage{pdflscape}
\usepackage{scalefnt}

\usepackage{lipsum}

\usepackage{listings}
\usepackage{tikz}
\usetikzlibrary{matrix,shapes,arrows,positioning}

\definecolor{codegreen}{rgb}{0,0.6,0}
\definecolor{codegray}{rgb}{0.5,0.5,0.5}
\definecolor{codepurple}{rgb}{0.58,0,0.82}
\definecolor{backcolour}{rgb}{0.95,0.95,0.92}

\lstdefinestyle{mystyle}{
    backgroundcolor=\color{backcolour},
    commentstyle=\color{codegreen},
    keywordstyle=\color{magenta},
    numberstyle=\tiny\color{codegray},
    stringstyle=\color{codepurple},
    basicstyle=\footnotesize,
    breakatwhitespace=false,
    breaklines=true,
    captionpos=b,
    keepspaces=true,
    numbers=left,
    numbersep=5pt,
    showspaces=false,
    showstringspaces=false,
    showtabs=true,
    tabsize=2}

\input cyracc.def

\newtheorem{rem}{\emph{Remark}}
\newtheorem{prop}{\emph{Proposition}}

\title{\textbf{\Large Theoretical results and modeling under the discrete Birnbaum-Saunders distribution}}

 \author{\bf \normalsize
  {{Filidor Vilca$^{1}$}}, Roberto Vila$^{2}$, Helton Saulo$^{2}$, Luis Sánchez$^{3}$ and Jeremias Le\~ao$^{4}$\\
  {\footnotesize Department of Statistics, Universidade Estadual de Campinas, Campinas, Brazil$^{1}$}
  \\
    {\footnotesize Department of Statistics, Universidade de Bras\'ilia, Bras\'ilia Brazil$^{2}$}\\
    {\footnotesize Institute of Statistics, Universidad Austral de Chile, Valdivia, Chile$^{3}$}\\
    {\footnotesize Department of Statistics, Universidade Federal do Amazonas, Manaus, Brazil$^{4}$}
  }

\date{\today}

\begin{document}

\maketitle

\begin{abstract}
In this paper, we discuss some theoretical results and properties of a discrete version of the Birnbaum-Saunders distribution. We present a proof of the unimodality of this model. Moreover, results on moments, quantile function, reliability and order statistics are also presented. In addition, we propose a regression model based on the discrete Birnbaum-Saunders distribution. The model parameters are estimated by the maximum likelihood method and a Monte Carlo study is performed to evaluate the performance of the estimators. Finally, we illustrate the proposed methodology with the use of real data sets.
\end{abstract}

\section{Introduction}

Despite the increasing number of works on discrete distributions in reliability, one can note that in many practical cases there is the need of more flexible distributions to model lifetime data. One way to develop new discrete distributions is by generating the discrete analogous of usual distributions for continuous lifetimes; see \cite{ALZAATREH2012589}. Some interesting discrete distributions are for example the discrete gamma distribution \citep{abouammohalhazzani:2015} and the discrete Weibull distribution \citep{vns:19}. It is  well known that in reliability, the continuous Birnbaum-Saunders (BS) distribution, proposed by \cite{bs:69a}, takes advantage than most continuous probability distributions, including the continuous gamma and Weibull distributions; see \cite{l:15}. The continuous BS distribution is a positively skewed model that is closely related to the normal distribution. Despite its origin in material fatigue, it has been considered in business, industry, insurance, inventory, quality control, among others; see, for example,
\cite{lp:08},
\cite{blsc:09},
\cite{acfls:10}, 
\cite{vslc:10},
\cite{plbl:12}, 
\cite{mbls:13}, 
\cite{rlwm:15}, 
\cite{wl:15},
\cite{lscc:11,lscb:14,lslm:14,lrsv:17},
\cite{sauloetal:2019},
\cite{dsls:16},
\cite{leaolst:18}, and
\cite{vslm:19}. Good recent references on the Birnbaum-Saunders distribution are \cite{l:15} and \cite{bk:19}. In particular, a positive random variable $T$ is said to follow a continuous BS distribution if its cumulative distribution function (CDF) is given by
\begin{eqnarray} \label{cdf}
F_T(t; \boldsymbol{\theta})= \Phi\left[ a(t;\boldsymbol{\theta})\right], \,  t>0,
\end{eqnarray}
where $\boldsymbol{\theta}=(\alpha, \beta)^{\top}$, with $\alpha>0$ and $\beta>0$ denoting the shape and scale parameters, respectively, $a(t;\boldsymbol{\theta})=\big(\sqrt{{t}/{\beta}} - \sqrt{{\beta}/{t}}\big)/\alpha$, and $\Phi[\cdot]$ is the standard normal CDF. This distribution   is usually denoted by $T\sim {\rm BS}(\boldsymbol{\theta})$. Even though the number of applications of the usual continuous BS distribution has been growing, there is a big number of applications where a discrete version of this distribution could be more appropriate. For example, to model the number of cycles or runs that a material or equipment supports before failing or breaking, the number of sessions of a treatment until the cure of a patient, or even the shelf life (in days) of a food product; see \cite{vns:19}.

In this paper, we study in more depth a discrete version of the continuous BS distribution, which was initially introduced by \cite{article-Subhradev}. The primary objectives of this paper are: (i) to discuss novel theoretical results and properties of this discrete BS (${\rm BS_d}$) distribution; and (ii) to introduce the corresponding regression model. The secondary objectives are: (i) to obtain the maximum likelihood estimates of the model parameters; (ii) to carry out Monte Carlo simulations to evaluate the performance of
the maximum likelihood estimators; and (iv) to discuss real data applications of the proposed methodology.

The rest of the paper proceeds as follows. In Section~\ref{sec:02}, we present the ${\rm BS_d}$ model and discuss some of its mathematical properties. Also, it is considered estimation of the model parameters based on  maximum likelihood method. In Section~\ref{sec:03}, a ${\rm BS_d}$ regression model is proposed, and the model parameter estimation is approached by using the maximum likelihood method. In Section ~\ref{sec:04}, we carry out Monte Carlo simulation studies to evaluate the performance of the estimators and we illustrate the proposed methodology with two real data sets. Finally, in Section~\ref{sec:05}, we make some concluding remarks.

\section{Discrete Birnbaum-Saunders distribution}\label{sec:02}
Before   defining   the  proposed     discrete      distribution,  we  present the probability density function (PDF) and quantile function of the continuous BS distribution. If $T\sim {\rm BS}(\boldsymbol{\theta})$, then its PDF is given by
\begin{align*}
f_T(t;\boldsymbol{\theta})= \phi\big[a(t;\boldsymbol{\theta})\big]a'(t;\boldsymbol{\theta}),  
\, t > 0,
\end{align*}
where $\phi[\,\cdot\,]$ is the PDF of  the standard normal distribution, $a(t;\boldsymbol{\theta})$ is  as  in (\ref{cdf})
and  $
a'(t;\boldsymbol{\theta}) =
(t+\beta)/(2t^{3/2}\alpha\beta^{1/2})$
is the derivative of $a(t;\boldsymbol{\theta})$ with respect to
$t$.  Moreover, the  $p$-th quantile of $T\sim {\rm BS}(\boldsymbol{\theta})$  is  given     by
\begin{equation} \label{quantile}
Q_p =
{\beta\over 4}\,
\left\{\alpha\Phi^{-1}(p)+\sqrt{[\alpha \Phi^{-1}(p)]^2+4}\right\}^2,
\end{equation}
where $p\in (0,1)$.

Now,  we  are ready   to present  a   discrete   random  variable   associated   to  the   positive  $T$   as  follows 
$
S=\lfloor T\rfloor,
$
where $\lfloor t\rfloor$ denotes  the largest integer contained in $t$.   
As the  set  of   all  possible   values  of $T$  is the set $(0,\infty)$, then  
$S=s$ iff $s< T \leq s+1, \, s=0,1,\ldots.$
Consequently, the probability mass function (PMF)  of  $S$  can  be     expressed   by 
\begin{align}\label{relation}
\mathbb{P}(S=s)= \mathbb{P}(s<T\leq s+1)
=
\begin{cases}
\Phi\big[a(1;\boldsymbol{\theta})\big], & \text{if} \ s=0,
\\[0,2cm]
\Phi\big[a(s+1;\boldsymbol{\theta})\big]-\Phi\big[a(s;\boldsymbol{\theta})\big], 
& \text{if} \ s=1, 2, \ldots,
\end{cases}
\end{align}
where $a(\cdot;\bm \theta)$ is as in \eqref{cdf}. We  can     show   that  $\sum_{s=0}^\infty \mathbb{P}(S=s)=\mathbb{P}(T>0)=1$, so  $\mathbb{P}(S=s)$ is  a  PMF.   On  the  other  hand, the CDF of $S$ is  given  by
\begin{align*}
{F}(s;\boldsymbol{\theta})
=
\mathbb{P}(S\leq s)
=
\Phi\big[a(\lfloor s\rfloor +1;\boldsymbol{\theta})\big]
\mathds{1}_{\{s\geq 0\}}.
\end{align*}
The     distribution   of  the   discrete   random  variable  $S$  will be   denoted  by $S\sim {\rm BS_d}(\boldsymbol{\theta})$ and  will be   called 
${\rm BS_d}$ distribution.

The reliability function (RF) and hazard rate function (HR) of $S\sim {\rm BS_d}(\boldsymbol{\theta})$ are, respectively, given by
\begin{align}\label{id-R}
R(s;\boldsymbol{\theta})
=
1-F(s;\boldsymbol{\theta})
=
1-\Phi\big[a({\lfloor s\rfloor}+1;\boldsymbol{\theta})\big] \mathds{1}_{\{s\geq 0\}},
\end{align}
\begin{align*} 
{H}(s;\boldsymbol{\theta})
=
{\mathbb{P}(S=s)\over \mathbb{P}(S=s)+R(s;\boldsymbol{\theta})}
&=
\begin{cases}
\Phi\big[a({1};\boldsymbol{\theta})\big], & \text{if} \ s=0,
\\[0,2cm]
{
	\Phi[a({s+1};\boldsymbol{\theta})]-\Phi[a({s};\boldsymbol{\theta})]
	\over 
	1-\Phi[a({s};\boldsymbol{\theta})]
}
& \text{if} \ s=1, 2, \ldots,
\end{cases}
\\[0,2cm]
&=
\begin{cases}
1-R(0;\boldsymbol{\theta}), & \text{if} \ s=0,
\\[0,2cm]
1-{R(s;\boldsymbol{\theta})\over R(s-1;\boldsymbol{\theta})},
& \text{if} \ s=1, 2, \ldots.
\end{cases}
\end{align*}
From the above identity, we have
\begin{align*}
R(s;\boldsymbol{\theta})
=
\prod_{y=0}^{s-1}
{R(y;\boldsymbol{\theta})\over R(y-1;\boldsymbol{\theta})}
=
\prod_{y=0}^{s-1}
\big[1-H(y;\boldsymbol{\theta})\big],
\quad s=0, 1, 2, \ldots,
\end{align*}
with the convention that $\prod_{y=0}^{-1} b_y=1$ and that $R(-1;\boldsymbol{\theta})=1$.

Figures \ref{fig:pmfs} and \ref{fig:hrs} displays different shapes of the ${\rm BS_d}$ PMF and HR for different choices of parameters. From these figures, we observe that the ${\rm BS_d}$ distribution possesses unimodal shapes for the PMF and HR.

\begin{figure}[h!]
\vspace{-0.25cm}
\centering
{\includegraphics[height=6.2cm,width=6.2cm]{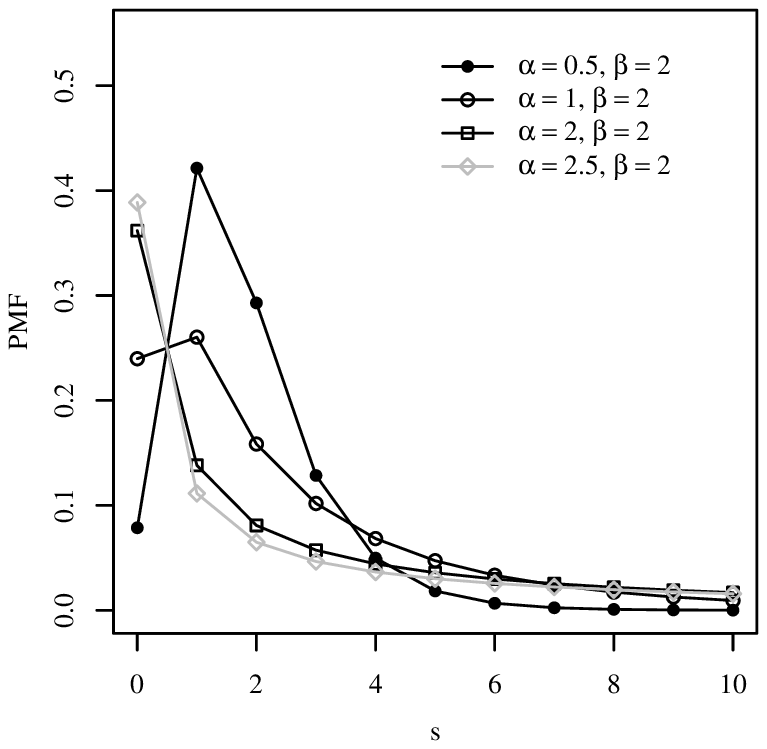}}
{\includegraphics[height=6.2cm,width=6.2cm]{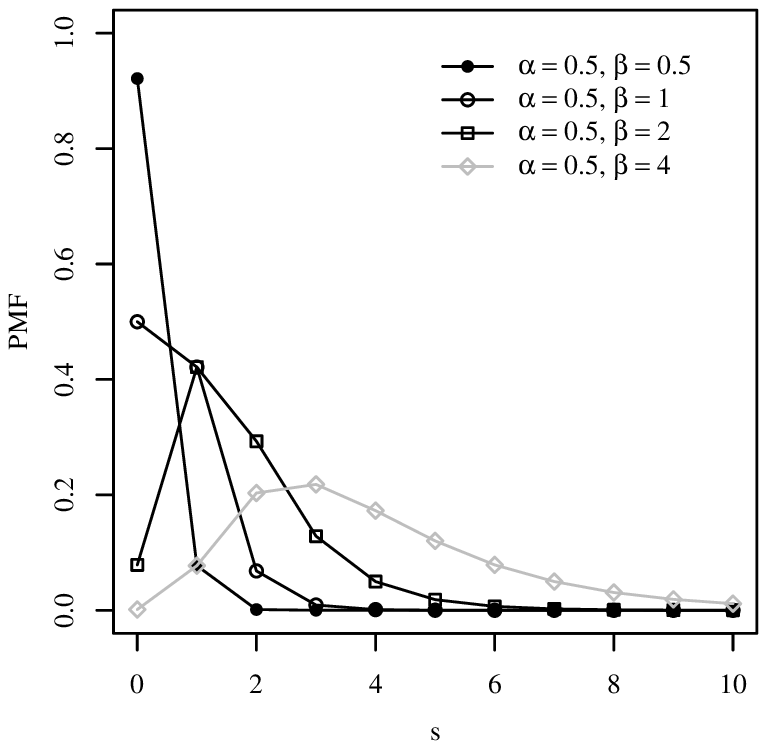}}
\vspace{-0.2cm}
\caption{Plots  of  the PMFs for several parameter values.}\label{fig:pmfs}
\end{figure}

\begin{figure}[h!]
\vspace{-0.25cm}
\centering
{\includegraphics[height=6.2cm,width=6.2cm]{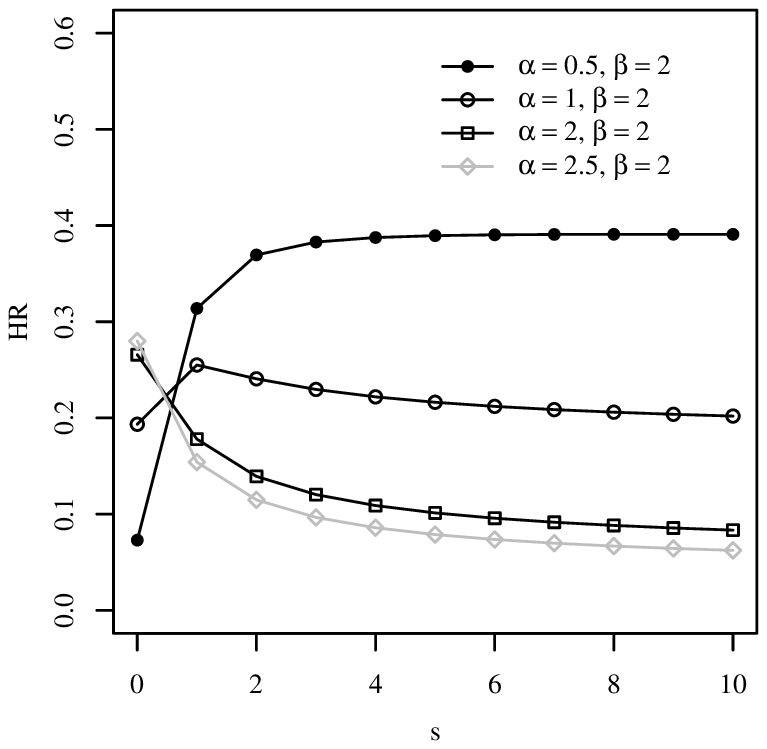}}
{\includegraphics[height=6.2cm,width=6.2cm]{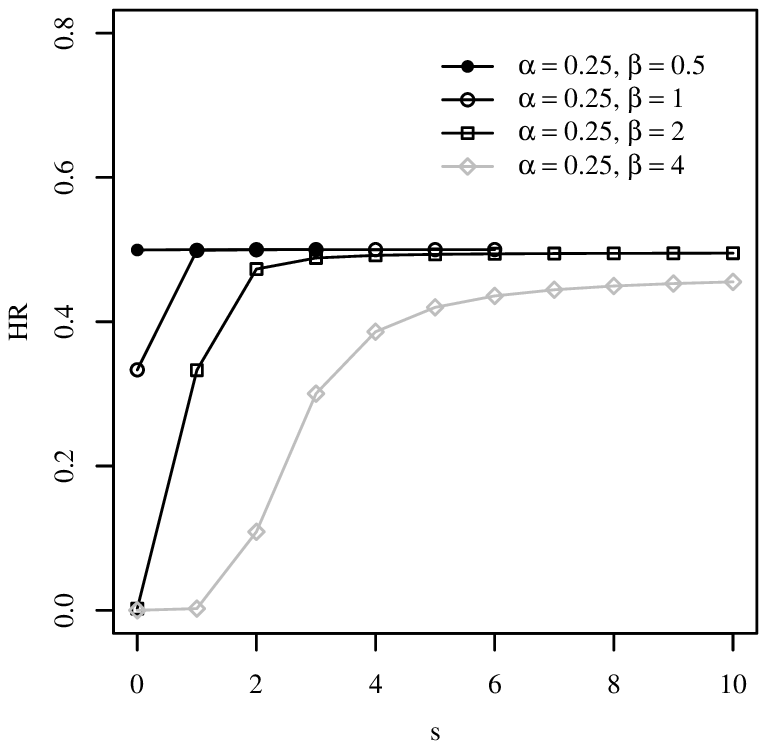}}
\vspace{-0.2cm}
\caption{$\text{BS}_{\text{d}}$ HRs for some parameter values.}\label{fig:hrs}
\end{figure}

\subsection{Properties}
We present some  properties  of  the   ${\rm BS_d}$   distribution,  many of the results can be easily derived from the definition  of the   ${\rm BS_d}$ distribution.
\begin{prop} \label{prop;1}
	If $S\sim {\rm BS_d}(\boldsymbol{\theta})$ and $t>0$, the following holds:
\begin{itemize}
\item[\rm (a)] 
$\sum_{s=0}^\infty \Phi\big[a({s+1};\boldsymbol{\theta})\big]
=
1+ \sum_{s=0}^\infty \Phi\big[a({s};\boldsymbol{\theta})\big];$
\item[\rm (b)] 
$\mathbb{P}(S\leq t)
= 
\mathbb{P}(S\leq \lfloor t\rfloor)
= 
\Phi\big[a({\lfloor t\rfloor+1};\boldsymbol{\theta})\big]
=
\mathbb{P}(T\leq \lfloor t\rfloor +1);$
\item[\rm (c)] 
$\mathbb{P}(S<t)= \mathbb{P}(T\leq \lfloor t\rfloor);$
\item[\rm (d)] 
$\mathbb{P}(S\geq t)= \mathbb{P}(T\geq \lfloor t\rfloor+1);$
\item[\rm (e)] 
$\mathbb{P}(S\leq \lfloor t\rfloor)= \mathbb{P}(T\leq\lfloor t\rfloor +1);$
\item[\rm (f)] 
$\mathbb{P}(S\geq \lfloor t\rfloor)= \mathbb{P}(T> \lfloor t\rfloor)$.
\end{itemize}
\end{prop}
%

\subsubsection{$\pmb{p}$-th quantile}
\begin{prop}\label{quantile}
Let $S\sim {\rm BS_d}(\boldsymbol{\theta})$ and  $ Q_p$   the   quantile  function    in (\ref{quantile}),  $p\in (0, 1)$.  
Then,
\begin{itemize}
\item[\rm (a)] If  $Q_p>0$ is a natural number, then  $Q_p-1$  is the  $p$-th quantile of the distribution of $S$;
\item[\rm (b)]
If  $Q_p>0$ is not a natural number, then $p$-th quantile of the distribution of $S$ can be represented by any value in the interval $\big[\lfloor Q_p\rfloor, \lfloor Q_p\rfloor+1\big)$.
\end{itemize}
\end{prop}
\begin{proof}
Since $Q_p$ is the $p$-th quantile for the continuous  random variable  $T\sim {\rm BS}(\boldsymbol{\theta})$,
from   Proposition \ref{prop;1}, Item (e),  we have
\[
\mathbb{P}(S< Q_p-1)
\leq
\mathbb{P}(S\leq Q_p-1)\stackrel{\rm (e)}{=} \mathbb{P}(T\leq Q_p)= p,
\quad S\sim {\rm BS_d}(\boldsymbol{\theta}),
\]
whenever $Q_p>0$ is a natural number.
So,  we  have that $Q_p-1$  is  the   $p$-th quantile of the distribution of $S$. This proves the first item.

Now, let  $t=Q_p>0$ be not a natural number.  From  Items (d) and (c) of Proposition \ref{prop;1} and from inequalities $\lfloor Q_p\rfloor\leq Q_p\leq \lfloor Q_p\rfloor+1$,
we have the following 
\begin{align*}
&\mathbb{P}(S< Q_p)\stackrel{\rm (d)}{=} \mathbb{P}(T\leq \lfloor Q_p\rfloor)
\leq \mathbb{P}(T\leq Q_p)=p,
\\
&\mathbb{P}(S\leq Q_p)\stackrel{\rm (c)}{=} \mathbb{P}(T\leq \lfloor Q_p\rfloor +1)\geq \mathbb{P}(T\leq Q_p) =p,
\end{align*}
and  consequently   
$\mathbb{P}(S< Q_p)\leq p \leq \mathbb{P}(S\leq Q_p)$. This will be true for any $\lfloor Q_p\rfloor \leq t< \lfloor Q_p\rfloor +1$. So, we have that,
at the percentage point $p$, the quantile for $S$ can be represented by any value in
$\big[\lfloor Q_p\rfloor, \lfloor Q_p\rfloor+1\big)$.
Thus we complete the proof.
\end{proof}
\begin{rem}
If
$p=0.5,$ then $Q_p=\beta$.
If $\beta$ is a natural number,
by Proposition \ref{quantile}-(a),
$m=\beta-1$ is   the   median  of  the   distribution  of  $S$.
Already, if $\beta$ is not a natural number, by Proposition \ref{quantile}-(b),
each $y\in\big[\lfloor\beta\rfloor, \lfloor\beta\rfloor+1\big)$ represents a median for $S$. 
\end{rem}

\subsubsection{Shape properties}
The next two  results   are  related   to   the unimodality  of  the  ${\rm BS_d}$ distribution.

\begin{prop}\label{prop-unimodal}
The ${\rm BS_d}$ distribution is unimodal. 
\end{prop}
\begin{proof}
Let  $T\sim {\rm BS}(\boldsymbol{\theta})$ be a random variable with continuous ${\rm BS}$ distribution. Let $f_T(t;\boldsymbol{\theta})$, $t>0$ be their respective PDF.
It is well-known that this distribution is unimodal (see Proposition 7 in \cite{10.1214/19-BJPS448}), then there exists a unique point $t_0>0$
such that its PDF satisfies the following inequalities:
\[
f_T(t;\boldsymbol{\theta})\geq  f_T(t-1;\boldsymbol{\theta}), \quad \text{for all} \ t\leq t_0,
\]
and
\[
f_T(t;\boldsymbol{\theta})\geq  f_T(t+1;\boldsymbol{\theta}), \quad \text{for all} \ t\geq t_0.
\]
If $s$ is a natural number such that $s\leq \lfloor t_0\rfloor-1$, then
\[
\mathbb{P}(S=s)
=
\int_{s}^{s+1} f_T(t;\boldsymbol{\theta}) \, {\rm d}t
\geq 
\int_{s}^{s+1} f_T(t-1;\boldsymbol{\theta}) \, {\rm d}t
=
\mathbb{P}(S=s-1),
\]
or equivalently,
\begin{align*}
\mathbb{P}(S=s)- \mathbb{P}(S=s-1)\geq 0 \quad \text{for all} \ s\leq \lfloor t_0\rfloor-1.
\end{align*}
Similarly, for $s\geq \lfloor t_0\rfloor+1$, we obtain
\[
\mathbb{P}(S=s+1)- \mathbb{P}(S=s)
\leq  0.
\]
It follows that $\big\{\mathbb{P}(S=s): s=0,1,2,\ldots\big\}$ is unimodal, whatever sign $\mathbb{P}(S=\lfloor t_0\rfloor)- \mathbb{P}(S=\lfloor t_0\rfloor-1)$ may have. 
\end{proof}
\begin{rem}
As a sub-product of the proof of Proposition \ref{prop-unimodal},
the mode of the ${\rm BS_d}$ distribution is $\lfloor t_0\rfloor$,
where 
$t_0$ is the mode of the corresponding continuous ${\rm BS}$ distribution.
\end{rem}
\begin{prop}
The ${\rm BS_d}$ distribution has a unique mode in the set $\big\{s=0,1,\ldots : 0 \leq s\leq \lfloor \beta\rfloor\big\}$.
\end{prop}
\begin{proof}
Proposition \ref{prop-unimodal} guarantees the uniqueness of mode. It remains to prove that the ${\rm BS_d}$ distribution $\mathbb{P}(S=s)$ is decreasing for all $s\geq \lfloor \beta\rfloor +1$. We prove this by comparing the continuous BS distribution with the corresponding ${\rm BS_d}$ distribution. Indeed, in Lemma 2.1. of \cite{10.1214/19-BJPS448} is proved that the PDF $f_T(t;\boldsymbol{\theta})$ of the continuous BS distribution is a decreasing function when $t>\beta$. However this extends to every $t\geq\beta$ because
$\big\{{{\rm d}\over {\rm d}t}
\log [f_T(t;\boldsymbol{\theta})]\big\}
\vert_{t=\beta}=a''(\beta;\boldsymbol{\theta})/a'(\beta;\boldsymbol{\theta}) <0$.
Hence, as a sub-product of the proof of Proposition \ref{prop-unimodal}, it follows that the ${\rm BS_d}$ PMF \eqref{relation} is decreasing for all $s\geq \lfloor \beta\rfloor +1$. Thus, we have completed the proof.
\end{proof}

\subsubsection{Order statistics}
\begin{prop}
If $S_1,\ldots,S_n$ is a sequence of independent and identically distributed random
	variables such that $S_1\sim {\rm BS_d}(\boldsymbol{\theta})$, 
	then,	
	the $i$th; $i = 1, 2,\ldots, n$; order statistic of the ${\rm BS_d}$
	distribution, denoted $S_{(i)}$, can be written as
	\[
	\mathbb{P}(S_{(i)}\leq  s)
	=
	\sum_{k=i}^{n}\sum_{j=0}^{n-k}\binom{n}{k}\binom{n-k}{j} (-1)^j
	{F}\big[a({\lfloor s\rfloor};\boldsymbol{\theta}); k+j\big],
	\quad s\geq  1,
	\]
	where 
	$
	{F}[\cdot; k+j]
	$
	denotes the CDF of the power normal distribution (PND).
	Different properties of the PND have been discussed by \cite{RG08}.
\end{prop}	
\begin{proof}
It is well-known that 
$\mathbb{P}(S_{(i)}\leq  s)
=
\sum_{k=i}^{n}\binom{n}{k}[F(s;\bm{\theta})]^k [R(s;\bm{\theta})]^{n-k}$, $s\geq  1$ 
(see Item (2.7) of \cite{shahbaz2016ordered}).
Using the Newton binomial formula and the definition of a PND, the proof follows.
\end{proof}
\begin{rem}
By using the identity
$
\sum_{k=i}^{n}\binom{n}{k}p^k(1-p)^{n-k}
=
i\binom{n}{i}\int_{0}^{p} t^{i-1}(1-t)^{n-i}\, {\rm d}t,
$
the distribution function of $S_{(i)}$ can also be written as
\[
{P}(S_{(i)}\leq  s)
=
{\Gamma(n+1)\over \Gamma(i)\Gamma(n-i+1)}\,
\int_{0}^{\Phi[a(\lfloor s\rfloor;\boldsymbol{\theta})]}
t^{i-1}(1-t)^{n-i}\, {\rm d}t,
\quad 
s\geq 1,
\]
where $\Gamma(\cdot)$ is the Gamma function.
\end{rem}
\subsubsection{Mean residual life function and variance residual life function}
Let $S\sim {\rm BS_d}(\boldsymbol{\theta})$, the mean residual life function (MRLF) and 
variance residual life function (VRLF) are defined   by 

$$
\mu_S(k)={\mathbb{E}(S-k\vert S\geq k)}
={\sum_{s=k}^{\infty} R(s;\boldsymbol{\theta})\over R(k-1;\boldsymbol{\theta})}
$$
and $$
 \sigma^2_S(k)={\rm Var}(S-k\vert S\geq k)
=
2{\sum_{s=k}^{\infty} s\, R(s;\boldsymbol{\theta})\over R(k-1;\boldsymbol{\theta})}
-(2k-1)\mu_S(k)-\mu^2_S(k)
,
$$
respectively, where $R(s;\boldsymbol{\theta})$ is given in \eqref{id-R} and $  k=0,1,\ldots,$.
%
\begin{prop}
Let $S\sim {\rm BS_d}(\boldsymbol{\theta})$  with   $\boldsymbol{\theta}$ belongs to the set 
$
\Theta=
\big\{\boldsymbol{\theta}_*\in (0,\infty)^2 :\ C(t;\boldsymbol{\theta}_*)=a''(t;\boldsymbol{\theta}_*)-
a(t;\boldsymbol{\theta}_*)[a'(t;\boldsymbol{\theta}_*)]^2>0, \ \ \forall \ t>0 \big\}.
$
Then,
\begin{itemize}
\item[\rm (a)]  $S$ has decreasing MRLF;
\item[\rm (b)] $S$ has increasing HR;
\item[\rm (c)] $S$ has decreasing VRLF,
\end{itemize}
whenever $\mathbb{P}(S>0)=1$.
\end{prop}
\begin{proof} 
For  $\boldsymbol{\theta}\in\Theta$,  we  have 
\begin{align*}
{{\rm d}^2\over {\rm d}t^2}\log\big[R(t;\boldsymbol{\theta})\big]
&=
-
{\phi\big[a(t;\boldsymbol{\theta})\big]\over 1-\Phi\big[a(t;\boldsymbol{\theta})\big]}\,
\left\{
{\phi\big[a(t;\boldsymbol{\theta})\big]
\big[a'(t;\boldsymbol{\theta})\big]^2  \over 1-\Phi\big[a(t;\boldsymbol{\theta})\big]}
+
C(t;\boldsymbol{\theta})
\right\}
<0 \quad \text{for all} \ t>0.
\end{align*}	
In other words, the function $\log[R(t;\boldsymbol{\theta})]$ is concave. This condition
implies that 
$
\log \big[R({t_1+t_2\over 2};\boldsymbol{\theta})\big]
\geq 
{1\over 2} \log \big[R(t_1;\boldsymbol{\theta})\big] 
+
{1\over 2} \log \big[R(t_2;\boldsymbol{\theta})\big]
$ 
or equivalently that 
$
\big[R({t_1+t_2\over 2};\boldsymbol{\theta})\big]^2
\ge 
R(t_1;\boldsymbol{\theta}) R(t_2;\boldsymbol{\theta})
$ 
for all 
$ 
t_1,t_2>0.
$ 
Hence, taking
$t_1=s+2$ and $t_2=s$ for $s=1,2,3,\ldots$, we have
\begin{align*}
\big[R(s+1;\boldsymbol{\theta})\big]^2
\ge R(s+2;\boldsymbol{\theta}) R(s;\boldsymbol{\theta})
\quad \Longleftrightarrow \quad
H(s+1;\boldsymbol{\theta})\leqslant H(s+2;\boldsymbol{\theta}).
\end{align*}
That is, 
$S$ has increasing hazard rate $H(\cdot;\boldsymbol{\theta})$.
Then, by Theorem 2.1 of 
\cite{doi:10.1080/03610926.2014.982827}, it follows that $S$ has decreasing mean residual life function. This proves the statement in Items (a) and (b).
Finally, the proof of Item (c) follows directly by combining Item (a) with Theorem 2.2 in
\cite{doi:10.1080/03610926.2014.982827}.
\end{proof}

\subsubsection{Moments properties}
\begin{prop}\label{exist-moments}
The distribution of a random variable $S$ with ${\rm BS_d}$ distribution has all moments.
\end{prop}
\begin{proof}
In Proposition \ref{prop-unimodal} is proved that $S$ has a
strongly unimodal distribution (see \cite{doi:10.1080/01621459.1971.10482273}, for a formal definition). 
Since that all strongly unimodal distribution have all moments (see Theorem 7 in \cite{doi:10.1080/01621459.1971.10482273}), the proof of the proposition follows. 
\end{proof}

\begin{prop}
If $S\sim {\rm BS_d}(\boldsymbol{\theta})$ is a random variable, for each natural number $r$, we have
\begin{eqnarray*}
\begin{array}{lllll}
 \textrm{\rm(a)} & \displaystyle
		\mathbb{E}(S^r) =  
		\sum_{s=0}^{\infty}\big[(s+1)^r-s^r\big]
		\left\{1-\Phi\big[a(s+1;\boldsymbol{\theta})\big]\right\};
\\[0,5cm]
\textrm{\rm(b)} 	& \displaystyle
		\mathbb{E}(S^r)
		=
		\sum_{s=0}^{\infty}\sum_{k=0}^{r}
		\sum_{i=0}^{r-k}
		\binom{r-k}{i}
		s^{k+i}\, 
		\left\{1-\Phi\big[a(s+1;\boldsymbol{\theta})\big]\right\};
\\[0,65cm]
\textrm{\rm(c)} 	& \displaystyle
		{\rm Var}(S)=
		2\sum_{s=0}^{\infty} s \left\{1-\Phi\big[a(s+1;\boldsymbol{\theta})\big]\right\}
		\\[0,2cm]
		&\qquad\quad\,+ \displaystyle
		\sum_{s=0}^{\infty} \left\{1-\Phi\big[a(s+1;\boldsymbol{\theta})\big]\right\}
		\bigg[
		1-\sum_{s=0}^{\infty} \left\{1-\Phi\big[a(s+1;\boldsymbol{\theta})\big]\right\}
		\bigg].
\end{array}	
\end{eqnarray*}
\end{prop}
\begin{proof}
The whole proof follows closely Proposition 2 of \cite{svpbm:21} and we present it for the sake of completeness. 
We emphasize that the statements of Items (a), (b) and (c) are  valid for any
discrete random variable $S$ with support $\{0,1,\ldots\}$.

By using the telescopic series $\sum_{x=0}^{i-1} [(x+1)^r-x^r]=i^r$, we have
	\begin{align*}
	\mathbb{E}(S^r)
	&=
	\sum_{i=0}^{\infty} \sum_{s=0}^{\infty} \mathds{1}_{\{s<i\}}  [(s+1)^r-s^r] \,
	\mathbb{P}(S=i)
	\\
	&=
	\sum_{s=0}^{\infty}  [(s+1)^r-s^r]  \sum_{i=0}^{\infty} \mathds{1}_{\{i>s\}} \,
	\mathbb{P}(S=i),
	\end{align*}
	where in the second equality we exchange the orders of the summations because
	\begin{align*}
	\sum_{s=0}^{\infty} \mathds{1}_{\{s<i\}}   \big|(s+1)^r-s^r\big| \,
	\mathbb{P}(S=i)
	= 
	\sum_{s=0}^{\infty} \mathds{1}_{\{s<i\}} \cdot [(s+1)^r-s^r] \,
	\mathbb{P}(S=i)
	=
	i^r \mathbb{P}(S=i),
	\end{align*}
	is finite for each $i=0,1,\ldots$; and because
	$\sum_{i=0}^{\infty}i^r \mathbb{P}(S=i)= \mathbb{E}(S^r)$ always exists (see Proposition \ref{exist-moments}).
	This proves Item (a).
	The second item follows by combining Item (a) with
	the polynomial identity $a^n-b^n = (a-b) \sum_{k=0}^{r} a^{r-k}b^k$ and the binomial expansion.
	Already, the proof of Item (c) is obtained by using Item (a) and simple algebraic manipulations.
\end{proof}
%

%
%
%
%

\subsection{Maximum likelihood estimation}
In this section, we discuss the  maximum likelihood estimation   for the unknown model parameters based on a random sample $S_1,S_2,\ldots,S_n$  from   $S\sim {\rm BS_d}(\boldsymbol{\theta})$,   with $\bm\theta=(\alpha, \beta)$.
Thus,  the  log-likelihood function for $\bm\theta$ is given by  
\begin{equation}\label{loglik}
l(\bm\theta)
=
\sum_{i=1}^{n}
\log\left\{ \Phi\big[a({s_{i}+1};\boldsymbol{\theta})\big]
-
\Phi\big[a({s_{i}};\boldsymbol{\theta})\big] \right\}.
\end{equation}
In order  to  obtain the   maximum likelihood  estimate of $\bm\theta$, we  have  the    score function given  by
$\dot{\bm l}(\bm\theta)=[\dot{l}_{\alpha}(\bm\theta),\dot{l}_{\beta}(\bm\theta)]^{\top}$, whose elements are given by
\begin{align} \label{first-der}
\dot{l}_{z}(\bm\theta)
&=
\sum _{i=1}^n 
\sum_{j=0}^{1}
(-1)^{j+1} \,
{\partial a({s_{i}+j};\bm\theta)\over \partial z}\,
\phi_{j}(s_{i}, \boldsymbol{\theta}),
\quad 
z\in\{\alpha,\beta\},
\end{align}
where
\begin{eqnarray*}
\frac{\partial
a(s_i;\bm\theta)}{\partial\alpha}&=&-\frac{1}{\alpha}a(s_i;\bm\theta), \, \frac{\partial
a(s;\bm\theta)}{\partial\beta}= -\frac{1}{2\alpha}\frac{1}{s^{1/2}\beta^{3/2}}(s+\beta)\, \, {\rm  and} \, \phi_{j}(s_{i}, \boldsymbol{\theta})= {\phi\big[a({s_{i}+j};\boldsymbol{\theta})\big]\over \mathbb{P}(S=s_i)}.
\end{eqnarray*}
The maximum likelihood estimate of  $\alpha$  and  $\beta$ can be
obtained solving the equations $\dot{l}_{\alpha}(\bm\theta)=0$ and $\dot{l}_{\beta}(\bm\theta)=0$  by an iterative procedure for non-linear optimization.  The Hessian matrix of ${l}(\bm\theta)$ is given by
$ \ddot{\bm l}(\bm\theta)
=
\big[
\frac{\partial^2 l(\bm\theta)}{\partial z\partial w}
\big]_{2\times 2},
$
for each $w,z\in\{\alpha,\beta\}$, where 
\begin{eqnarray}\label{sec-der-a}
\ddot{l}_{wz}(\bm\theta)&=&\sum _{i=1}^n 
\sum_{j=0}^{1}
(-1)^{j+1}
{\partial^2 a({s_{i}+j};\bm\theta)\over\partial w \partial z} \,
{\phi_{j}(s_{i},\boldsymbol{\theta})}  \nonumber \\
&-&
\sum _{i=1}^n 
\sum_{j=0}^{1}
(-1)^{j+1}
a({s_{i}+j};\boldsymbol{\theta})\,
{\partial a({s_{i}+j};\bm\theta)\over\partial w } \,
{\partial a({s_{i}+j};\bm\theta)\over\partial z }\,
{\phi_{j}(s_{i},\boldsymbol{\theta})} \nonumber \\
&-&
\sum _{i=1}^n 
\sum_{j=0}^{1}
(-1)^{j+1}
{\partial a({s_{i}+j};\bm\theta)\over\partial z }\,
{\phi_{j}\big[s_{i},\boldsymbol{\theta}\big]}
\sum_{k=0}^{1} (-1)^{k+1}\,
{\partial a({s_{i}+k};\bm\theta)\over\partial w } \,
{\phi_{k}(s_{i},\boldsymbol{\theta})},
\end{eqnarray}
where $\phi_{j}(s_{i}, \boldsymbol{\theta})$ is  as  in (\ref{first-der}) and  the second-order partial derivatives of $a(\cdot;\bm\theta)$, with respect to the parameters, are given by
\begin{eqnarray*}
\frac{\partial^2
a(s_i;\bm\theta)}{\partial\alpha\partial\alpha}&=&\frac{2}{\alpha^2}a(s_i), \,\, \,\frac{\partial^2
a(s_i;\bm\theta)}{\partial\beta\partial\alpha}=-\frac{1}{\alpha}\frac{\partial
a(s_i;\bm\theta)}{\partial\beta}, \, \,\, \frac{\partial^2
a(s_i;\bm\theta)}{\partial\beta\partial\beta}=\frac{1}{4\alpha s_i^{1/2}\beta^{5/2}}(3s_i+\beta).
\end{eqnarray*}

\section{Discrete Birnbaum-Saunders regression model}\label{sec:03}

In the context of count data, the  ${\rm BS_d}$ distribution  may   be   an  interesting  alternative distribution to   usual   discrete  distributions  or   to    those discrete    distributions     have    been  derived from     continuous   distributions. Then,  for  the  ${\rm BS_d}$   distribution   we     are     also  going   to    consider   its    associated  regression     model,  which    will   be   the  goal of this  part    of  the  study.  The   associated   ${\rm BS_d}$  regression  model that   we  are going   to introduce    is   inspired  by      continuous   BS   regression  model developed  by  \cite{balazhu:15},  where they considered the  scale  parameter    depending  on covariates.

Suppose that  we observe   independent failure times  $S_1,\ldots, S_n$,  such   as 
\begin{eqnarray}\label{regeq01}
 S_{i}\sim {\rm BS_d}(\boldsymbol{\theta}_{i}), 
 \end{eqnarray}
where $\boldsymbol{\theta}_{i}=(\alpha,\beta_{i})$, $i=1,\ldots,n$.  The   distribution    depends   on  covariates  $\bm{x}_i= (x_{1i},\ldots, x_{pi})^{\top}$  associated   with  $\beta_i$   thought   $\beta_i=\exp(\bm{x}_i^\top \bm{\eta})$, with $\bm{\eta}=(\eta_0,\eta_1,
\ldots, \eta_p)^\top$ being a vector of unknown parameters. The corresponding PMF associated with (\ref{regeq01}) is
\begin{eqnarray*}
\mathbb{P}(S_{i}=s_{i}) 
&= &  \Phi\big[a(s_i;\boldsymbol{\theta}_i)\big]\mathds{1}_{\{s_i=0\}} +
\left\{
\Phi\left(a(s_i+1;\boldsymbol{\theta}_i)\right)
- \Phi\left( a(s_i;\boldsymbol{\theta}_i)\right)\right\}\mathds{1}_{\{s_i\geqslant 1\}},
\end{eqnarray*}
$i=1,\ldots,n$.

\subsection{Maximum likelihood estimation}
The  log-likelihood function for ${\bm\theta}=(\alpha, \bm{\eta}^{\top} )^{\top}$ is given by
\begin{eqnarray}\label{loglikreg}
l(\bm\theta)&=&
\sum_{i=1}^{n}
\log\left\{ \Phi\big[a(s_i;\boldsymbol{\theta}_i)\big]\mathds{1}_{\{s_i=0\}} +
\left\{
\Phi\left(a(s_i+1;\boldsymbol{\theta}_i)\right)
- \Phi\left( a(s_i;\boldsymbol{\theta}_i)\right)\right\}\mathds{1}_{\{s_i\geqslant 1\}} \right\}.
\end{eqnarray}

Then, the first derivatives of the log-likelihood function \eqref{loglikreg},
with ${\bm\theta}=(\alpha, \bm{\eta}^{\top} )^{\top}$, can be written as
\begin{eqnarray*}
\dot{l}_{u}(\bm\theta)
&=& \sum _{i=1}^n  \sum_{j=0}^{1}
(-1)^{j+1} \,
{\partial a({s_{i}+j};\bm\theta)\over \partial z_i}\,
{\partial z_i\over \partial u}\,\phi_{j}(s_{i},\boldsymbol{\theta})
\quad 
u\in\{\alpha, \bm{\eta}\}, z_i\in\{\alpha,\beta_i\},
\end{eqnarray*}
where
$\dot{l}_{u}(\bm\theta)={\partial l(\bm\theta)}/{\partial u}$, and
${\partial a({s}_i;\bm\theta)/ \partial z_i}$ is as in \eqref{first-der}.  Specifically 
\begin{eqnarray*}
\dot{l}_{\alpha}(\bm\theta)&=&\sum _{i=1}^n \sum_{j=0}^{1}
(-1)^{j+1} \, {\partial a({s_{i}+j};\bm\theta)\over \partial \alpha}\,
\phi_{j}(s_{i},\boldsymbol{\theta}), \\
\dot{l}_{\eta}(\bm\theta)&=&\sum _{i=1}^n \sum_{j=0}^{1}
(-1)^{j+1} \, {\partial a({s_{i}+j};\bm\theta)\over \partial \beta_i}\,
\phi_{j}(s_{i},\boldsymbol{\theta}) {\partial \beta_i\over \partial \bm{\eta}},
\end{eqnarray*}
where ${\partial\beta_i/\partial \bm{\eta}}=\beta_i \bm{x}_i; i=1,\ldots,n.$  From  the likelihood equations
$\dot{l}_{\alpha}(\bm\theta)=0$  and  $\dot{l}_{\eta}(\bm\theta)=0$,  we  can  see  that  there  is no closed-form solution to the maximization problem, so  we   implement  two algorithms in software \texttt{R} to find the maximum likelihood estimates of $\alpha$, $\beta$ and $\eta_i$, $i=0,\ldots, p$, by using the function \texttt{optim()}; see \cite{rCoreTeams:20}. These procedures are evaluated and used in the next section.

Furthermore, the Hessian matrix of $l(\bm\theta)$  is given by
\[
\big[\ddot{l}_{vu}(\bm\theta)\big]_{p\times p}
= \begin{bmatrix}
\frac{\partial^2 l(\bm\theta)}{\partial \alpha^2}
&  \frac{\partial^2 l(\bm\theta)}{\partial \alpha \partial  \bm{\eta}^{\top}}
\\[0,2cm]
\frac{\partial^2 l(\bm\theta)}{\partial \bm{\eta}\partial \alpha}
&  \frac{\partial^2 l(\bm\theta)}{\partial \bm{\eta}\partial \bm{\eta}^{\top}}
\end{bmatrix},
\]
where, for each $v,u\in\{\alpha, \bm{\eta} \}$ and $w_i,z_i\in\{\alpha,\beta_i\}$, the  elements  of  the  Hessian   matrix are   given  by 
\begin{eqnarray*}
\ddot{l}_{vu}(\bm\theta)
&=& \sum _{i=1}^n 
\sum_{j=0}^{1}
(-1)^{j+1}
\left[
{\partial^2 a({s_{i}+j};\bm\theta)\over\partial w_i \partial z_i} \,
{\partial z_i\over \partial u} \,
{\partial w_i\over \partial v}
+
{\partial a({s_{i}+j};\bm\theta)\over\partial z_i }\,
{\partial^2 z_i\over\partial v \partial u}
\right]
\phi_{j}(s_{i},\boldsymbol{\theta}) \\
&-&\sum _{i=1}^n 
\sum_{j=0}^{1}
(-1)^{j+1}
a({s_{i}+j};\boldsymbol{\theta})\,
{\partial a({s_{i}+j};\bm\theta)\over\partial w_i } \,
{\partial a({s_{i}+j};\bm\theta)\over\partial z_i }\,
{\partial z_i\over \partial u} \,
{\partial w_i\over \partial v}\,
\phi_{j}(s_{i},\boldsymbol{\theta})
\\
&-&\sum _{i=1}^n 
\sum_{j=0}^{1}
(-1)^{j+1}
{\partial a({s_{i}+j};\bm\theta)\over\partial z_i }\,
\phi_{j}(s_{i},\boldsymbol{\theta})
\sum_{k=0}^{1} (-1)^{k+1}\,
{\partial a({s_{i}+k};\bm\theta)\over\partial w_i } \,
{\partial z_i\over \partial u} \,
{\partial w_i\over \partial v}\,
\phi_{k}(s_{i},\boldsymbol{\theta}),
\end{eqnarray*}
where ${\partial^2 a({s}_i;\bm\theta)/\partial w_i\partial z_i}$ is as in \eqref{sec-der-a}
and $$ \frac{\partial
\beta_i}{\partial \bm{\eta}}= \beta_i\bm{x}_i \, \,\, {\rm and} \, \, \,  \frac{\partial^2
\beta_i}{\partial \bm{\eta}\partial \bm{\eta}^{\top}}= \beta_i\bm{x}_i \bm{x}_i^{\top} $$

{\color{black}Again, note that the equation $\dot{l}_{u}(\bm\theta) = \bm 0$ does not provides analytic solutions for $\widehat{\alpha}$ and $\widehat{\eta_j}$, $j=0,\ldots, p$. Therefore, we have implemented two algorithms in software \texttt{R} to find the maximum likelihood estimates of $\alpha$ and $\eta_i$, $i=0,\ldots, p$, by using the function \texttt{optim()}; see \cite{rCoreTeams:20}. These procedures are evaluated and used in the next section.}

\section{Numerical evaluation}\label{sec:04}

In this section we carry out a simulation study to evaluate the performance of both the maximum likelihood estimators and residuals. Moreover, we analyse two real data sets. All numerical evaluations were done in the \texttt{R} software; see \cite{rsoftware:19}. The \texttt{R} codes are available upon request from the authors.

\subsection{Simulation}
 We first evaluate the performance of the maximum likelihood estimators for the $S\sim {\rm BS_d}$ model. Then, we consider a ${\rm BS_d}$ regression model where the parameter $\beta$ is associated with a covariate, that is,
\begin{eqnarray}\label{eqsim1}
\beta_{i}=\exp(\eta_{0}+\eta_{1}x_i)
\quad i=1,\ldots,n.
\end{eqnarray}
In \eqref{eqsim1}, the covariate values were randomly generated from the uniform distribution in the interval (0,1). The simulation scenario considers: sample size $n \in \{10, 50,150, 400\}$ and the values of the shape parameter as $\alpha \in \{0.50,1.50,1.50,3.00\}$ , with $1,000$ Monte Carlo replications for each sample size. The values of $\alpha$ have been chosen to cover the performance under low, moderate and high skewness. The $\text{BS}_{\text{d}}$  samples were generated using the Proposition \ref{quantile}.

The maximum likelihood estimation results for the ${\rm BS_d}$ model are presented in Table~\ref{tab:mc1}. We report the following sample statistics for the maximum likelihood estimates: empirical bias and mean squared error (MSE). Note that the results in Table~\ref{tab:mc1} allows us to conclude that, as the sample size increases, the bias and MSE of the estimators $\widehat{\alpha}$ and $\widehat{\beta}$ decrease, indicating that they are asymptotically unbiased, as expected.

\begin{table}[!ht]
 \tiny
\centering
\caption{Simulated values of biases (MSEs within parentheses) of the estimators of the ${\rm BS_d}$ model ($\beta=2$).}
\label{tab:mc1}
\renewcommand{\arraystretch}{1.8}
\resizebox{\linewidth}{!}{
\begin{tabular}{clrrrrrrrrrrrrrrrrrr}
\hline
        &&\multicolumn{2}{c}{$n=10$} && \multicolumn{2}{c}{$n=50$}    \\ \\[-0.6cm]
    && \multicolumn{1}{c}{$\widehat{\alpha}$}   &\multicolumn{1}{c}{$\widehat{\beta}$} &&\multicolumn{1}{c}{$\widehat{\alpha}$}&   \multicolumn{1}{c}{$\widehat{\beta}$}\\[-0.1cm] \hline 
0.5 &&   $-$0.0406(0.0213) &   0.0281(0.1052) &&   $-$0.0075(0.0035) &   0.0038(0.0199)\\[-0.2cm]
1.5 &&   $-$0.1462(0.1710) &   0.1138(0.7972) &&      0.0047(0.0488) &   0.0207(0.1785)\\[-0.2cm]
2.5 &&   $-$0.4642(0.6149) &   0.8073(2.6281) &&   $-$0.1262(0.1424) &   0.2500(0.3770)\\[-0.2cm]
3.0 &&   $-$0.6597(0.9420) &   1.2008(4.5642) &&   $-$0.2351(0.2279) &   0.4059(0.5581)\\[-0.2cm]
    &&\multicolumn{2}{c}{$n=150$} && \multicolumn{2}{c}{$n=400$}    \\  \\[-0.5cm]
0.5 &&   $-$0.0013(0.0011) &     0.0017(0.0069) &&   $-$0.0002(0.0004) &  $-$0.0013(0.0026)\\[-0.2cm]
1.5 &&   $-$0.0021(0.0157) &  $-$0.0018(0.0566) &&   $-$0.0007(0.0062) &  $-$0.0028(0.0215)\\[-0.2cm]
2.5 &&   $-$0.0372(0.0531) &     0.1024(0.1312) &&   $-$0.0135(0.0255) &     0.0319(0.0561)\\[-0.2cm]
3.0 &&   $-$0.0998(0.0918) &     0.1807(0.1908) &&   $-$0.0180(0.0433) &     0.0410(0.0765)\\
         \hline
\end{tabular}
}
\end{table}

Table~\ref{tab:mc2} reports the simulation results for the ${\rm BS_d}$ regression model. A look at the results in Table \ref{tab:mc2} allows us to conclude that, as the sample size increases, the empirical bias and MSE decrease, as expected. Moreover, we note that, as the value of the parameter $\alpha$ increases, the performances of the estimators of $\beta_{0}$,
$\beta_{1}$ and $\alpha$, deteriorate. 

\begin{table}[!ht]
\centering
\caption{Simulated values of biases (MSEs within parentheses) of the estimators of the ${\rm BS_d}$ regression model ({\color{black}$\eta_{0}=0.2$ and $\eta_{0}=1.5$}).}
\label{tab:mc2}
\renewcommand{\arraystretch}{1.5}
\resizebox{\linewidth}{!}{
\begin{tabular}{clrrrrrrrrrrrrrrrrrr}
\hline
    &&\multicolumn{3}{c}{$n=10$} && \multicolumn{3}{c}{$n=50$}    \\ \\[-0.5cm]
    && \multicolumn{1}{c}{$\widehat{\eta}_{0}$}   &\multicolumn{1}{c}{$\widehat{\eta}_{1}$} &\multicolumn{1}{c}{$\widehat{\alpha}$}
    &&\multicolumn{1}{c}{$\widehat{\eta}_{0}$}   &\multicolumn{1}{c}{$\widehat{\eta}_{1}$} &\multicolumn{1}{c}{$\widehat{\alpha}$}\\[-0.1cm] \hline 
0.5 &&      0.0122(0.1581)&$-$0.0065(0.4129)&$-$0.0586(0.0155) &&   0.0046(0.0235)&$-$0.0043(0.0665)&$-$0.0104(0.0027)\\[-0.2cm]
1.5 &&   $-$0.0457(1.1930)&0.0621(3.0408)&$-$0.1665(0.1672)    &&$-$0.0072(0.1611)&0.0108(0.4365)&$-$0.0268(0.0279)\\[-0.2cm]
2.5 &&   $-$0.2260(3.6180)&0.1322(8.2195)&$-$0.0438(3.0714)    &&$-$0.0473(0.3325)&0.0326(0.8433)&$-$0.0119(0.1108)\\[-0.2cm]
3.0 &&   $-$0.3103(5.2066)&0.1576(11.5868)&0.2436(23.5472)     &&$-$0.0632(0.4081)&0.0242(0.9987)&0.0156(0.2055)\\[-0.1cm]
    &&\multicolumn{3}{c}{$n=150$} && \multicolumn{3}{c}{$n=400$}    \\ \\[-0.5cm]
0.5 &&   $-$0.0003(0.0082)&0.0023(0.0212)&$-$0.0027(0.0009)    &&$-$0.0007(0.0030)&0.0011(0.0078)&$-$0.0009(0.0004)\\[-0.2cm]
1.5 &&   $-$0.0087(0.0534)&0.0160(0.1354)&$-$0.0072(0.0092)    &&$-$0.0043(0.0190)&0.0039(0.0499)&$-$0.0015(0.0036)\\[-0.2cm]
2.5 &&   $-$0.0195(0.1057)&0.0208(0.2493)&$-$0.0034(0.0359)    &&$-$0.0128(0.0356)&0.0107(0.0879)&   0.0024(0.0133)\\[-0.2cm]
3.0 &&   $-$0.0270(0.1289)&0.0209(0.2983)&0.0063(0.0590)       &&$-$0.0166(0.0433)&0.0143(0.1003)&0.0050(0.0230)\\

         \hline
\end{tabular}
}
\end{table}

\subsection{Examples}

The ${\rm BS_d}$ distribution and its regression model proposed in Section \ref{sec:03} are now used to analyze two data sets. In the first case, the objective is to fit the BS$_{\rm d}$ distribution to data corresponding to biaxial fatigue-life of $n = 46$ metal specimens (in cycles) until failure; this data set can be found in \cite{rieck:1989}. In the second example, we fit the proposed regression model to data on the fatigue-life (in cycles $\times 10^{-3}$) of concrete specimens (response variable $Y$), where the covariate is the ratio of applied stress causing failure (covariate $x$); see \cite{mills:1997}. In this second data set, the number of observations is $n=45$.

\paragraph{Case study 1: Metal specimens}  A descriptive summary of this data provides 
the following sample values: 566(median); 943.065(mean); 1110.934(standard deviation); 117.8(coefficient of variation); 2.204(coefficient of skewness); 4.682(coefficient of kurtosis), whereas their minimum and maximum times are 125 and 5046, respectively. {The histogram shown in Figure \ref{fig:histogram} and the value of the coefficient of skewness support the assumption that these data follow an asymmetrical distribution. We have assumed different discrete asymmetrical distributions to describe this data set, including the Weibull, gamma, log-normal, log-Student-$t$, and log-power-exponential (log-PE) distributions; see \cite{nakagawaozaki:1975}, \cite{abouammohalhazzani:2015}, and \cite{svpbm:21}. Table \ref{criteria} presents the Akaike (AIC) and Bayesian (BIC) information criteria. The results of Table \ref{criteria} reveal that the BS$_{\rm d}$ model provides better adjustment than the other models based on the values of AIC and BIC. The estimates and standard errors (in parenthesis) for the BS$_{\rm d}$ model are $\widehat{\alpha}=1.0840 (0.1130)$ and $\widehat{\beta}=595.1987 (81.6782)$, and the fitted PMF is also shown in Figure \ref{fig:histogram}.}

\begin{center}
\begin{table}[h]	
    \caption{{Values of AIC and BIC for different discrete asymmetrical distributions.}}
    \resizebox{16cm}{!}{     
	\begin{tabular}{ccccccccc}
	\hline
{criterion} & {Weibull} & {Gamma} & {log-normal} & {log-Student-$t$} & {log-PE}  & {BS$_{\rm d}$} \\
\hline AIC  &  726.1691 &  726.1692 & 718.8391   & 719.7321       & 715.8479        & 714.8548              \\
BIC        &  729.8264  & 729.8265 & 724.3250    & 725.2181       & 721.3338          & 718.5121          \\
	\hline
	\end{tabular}}
   \label{criteria}
\end{table}
\end{center}

\vspace{-0.5cm}

\begin{figure}[h]
	\centering
	\includegraphics[width=0.4\linewidth]{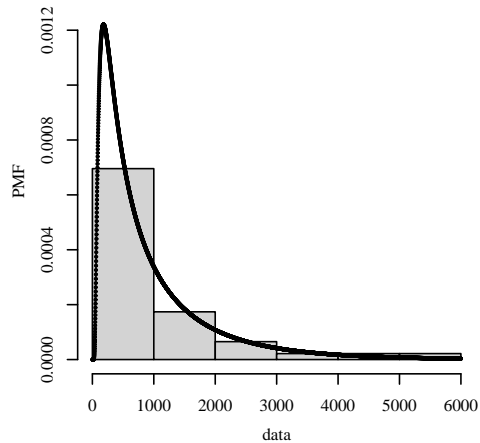}
	\caption{histogram for data of metal specimens.}
	\label{fig:histogram}
\end{figure}

\paragraph{Case study 2: Concrete specimens}  The number of cycles until failure is expected to increase inversely with the ratio of applied stress causing failure. The postulated model is given by
\begin{equation*}
	\beta_i = \exp(\eta_0 + \eta_1 x), \quad Y_i \sim \textrm{BS}_d(\alpha, \beta_i),
\end{equation*}	
for $i=1,\ldots,45$. The maximum likelihood estimates and standard errors (in parenthesis) for $\alpha$, $\eta_0$ and $\eta_1$ are $\widehat{\alpha}=0.4966(0.0641)$, $\widehat{\eta}_0=27.4913(3.2530)$ and $\widehat{\eta}_1 = -23.9647(3.5146)$, respectively. Figure \ref{residuals} presents the QQ plots with envelope of the generalized Cox-Snell and randomized quantile residuals for the ${\rm BS_d}$ regression model; see \cite{sauloetal:2019}. Note that all points are inside the bands and around the $y=x$ line, demonstrating a very good fit of the proposed model.

\begin{figure}[h!]
	\centering
	\subfloat[]{
		\includegraphics[width=0.46\linewidth,height=0.35\linewidth]{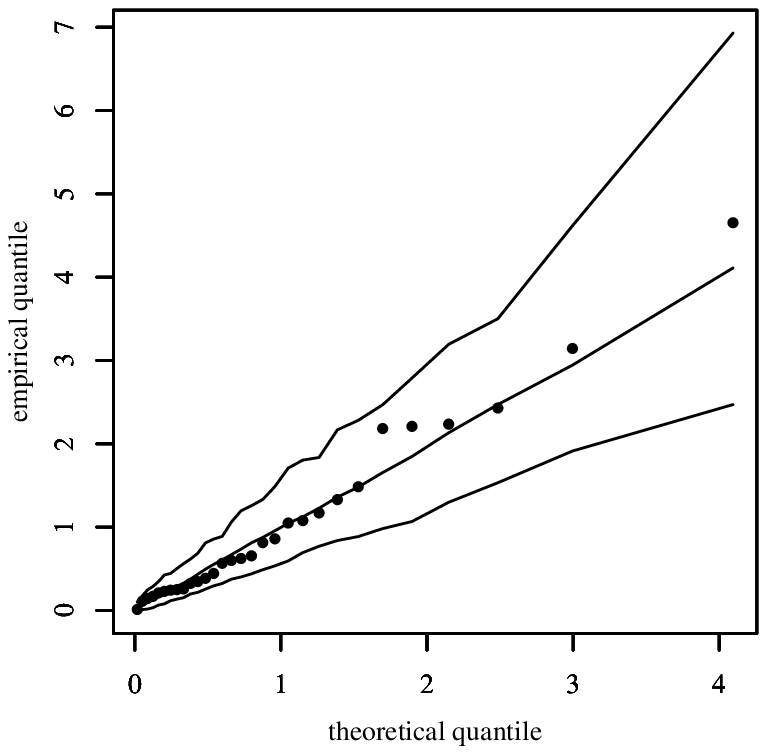}} \hspace{0.3cm}
	\subfloat[]{
		\includegraphics[width=0.46\linewidth,height=0.35\linewidth]{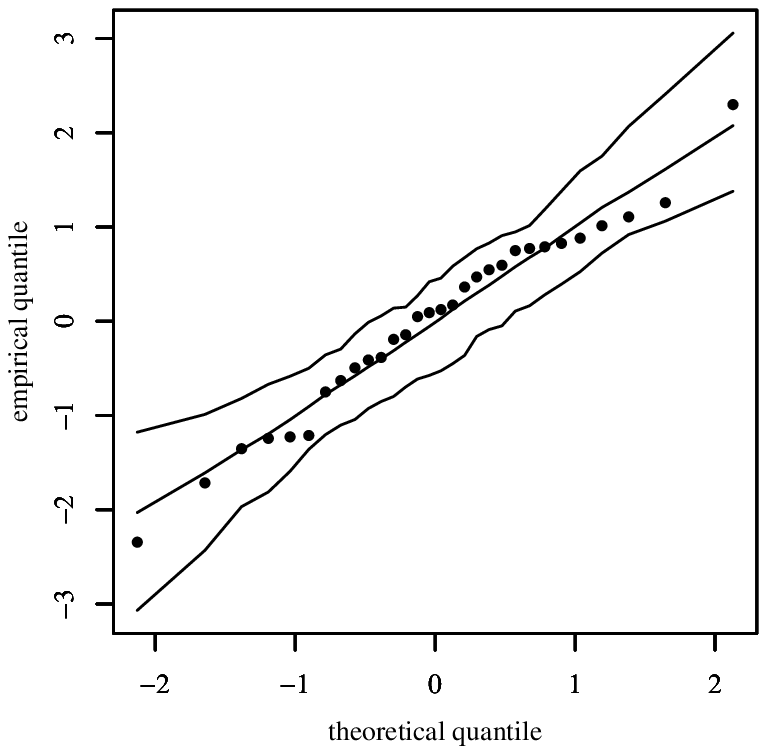}} 
	\caption{QQ plots with envelope for the generalized Cox-Snell (a) and randomized quantile (b) residuals for the concrete specimens data.}
	\label{residuals}
\end{figure}

\section{Concluding remarks}\label{sec:05}

The continuous Birnbaum-Saunders distribution has been widely used in several areas, besides being an alternative to the Weibull and gamma distributions. However, in many practical problems, the use of discrete distributions is more appropriate. In this sense, we have studied a discrete version of the Birnbaum-Saunders distribution. Some important properties have been presented, such as moments, quantile function and reliability. We have presented a formal proof concerning the unimodality property of discrete Birnbaum-Saunders distribution. In addition, we have proposed a new discrete Birnbaum-Saunders regression model. Monte Carlo simulations have been carried out to evaluate the behaviour of
the maximum likelihood estimators. Two examples with real data have illustrated the proposed methodology. The results are seen to be quite favorable to the discrete Birnbaum-Saunders distribution as well as its regression model in terms of model fitting.

\end{document}